\documentclass[runningheads]{llncs}

\usepackage{mathptmx,hyperref}   
\usepackage{helvet} 
\usepackage{courier}
\usepackage{amsmath,amsfonts}
\usepackage{graphicx}
\usepackage[bottom]{footmisc}          
    
\usepackage{mathptmx}         
\usepackage{helvet}            
\usepackage{courier}            
\usepackage{type1cm}             
\usepackage{makeidx}             
\usepackage{graphicx}          
\usepackage{multicol}          
\usepackage[bottom]{footmisc}  
     
\usepackage{epsfig}     
\usepackage{psfrag,rotating}     
\usepackage{amssymb,floatflt,enumerate} 
\usepackage{amsmath,amscd,psfrag,leqno}  
\usepackage[mathscr]{eucal}  
\usepackage{shadethm}           
\usepackage{graphicx,subfigure}   
\usepackage{overpic,contour}       
\contourlength{0.3mm}  
 
\def\Beweisende{\square}            
\def\BewEnde{\hfill{\Beweisende}}

\def\phm{{\hphantom{-}}}

\def\CC{{\mathbb C}}

\newcommand{\kreuz}{\!\times}

\newcommand{\go}[1]{{\sf #1}}

\DeclareMathOperator{\sign}{sign}

\begin{document}

\title{Generalizing continuous flexible Kokotsakis belts \\ of the isogonal type\thanks{Dedicated to my newborn son and his mother on the occasion of his birth.}}
\author{Georg Nawratil\inst{}\orcidID{0000-0001-8639-9064} 
%\and Second Author\inst{2,3}\orcidID{1111-2222-3333-4444} 
%\and Third Author\inst{3}\orcidID{2222--3333-4444-5555}
}
\authorrunning{G. Nawratil}
\institute{Institute of Discrete Mathematics and Geometry \& \\ 
	Center for Geometry and Computational Design, \\TU Wien, Austria, 
	\email{nawratil@geometrie.tuwien.ac.at}\\
	\url{https://www.geometrie.tuwien.ac.at/nawratil/}
}
\maketitle              % typeset the header of the contribution
\begin{abstract}
Kokotsakis studied the following problem in 1932: Given is a rigid closed polygonal line (planar or non-planar), which is 
surrounded by a polyhedral strip, where at each polygon vertex three faces meet. Determine the 
geometries of these closed strips with a continuous mobility. 
On the one side, we generalize this problem by allowing the faces, which are adjacent to polygon line-segments, to be skew; i.e to be non-planar. 
But on the other side, we restrict to the case where the four angles associated with each polygon vertex
fulfill the so-called isogonality condition that both pairs of opposite angles are equal or supplementary. 
In more detail, we study the case where the polygonal line is a skew quad, as this corresponds to a $(3\times 3)$ 
building block of a so-called V-hedra composed of skew quads. 
The latter also gives a positive answer to a question  posed by Robert Sauer in his book of 1970 whether continuous flexible skew quad surfaces exist. 

\end{abstract}
\section{Introduction}\label{sec:intro}

Let us consider a so-called Kokotsakis belt as described in \cite{kokotsakis}, which is illustrated in Fig.\ \ref{fig1}a. 
In general these loop structures are rigid, thus continuous flexible ones possess a so-called 
overconstrained mobility. Kokotsakis himself formulated the problem for general rigid closed polygonal lines $\go p$, but in fact he 
only studied flexible belts with planar polygons $\go p$ in \cite{kokotsakis}. Planarity was only not assumed 
in the study of necessary and sufficient conditions for infinitesimal flexibility (see also Karpenkov \cite{karpenkov}). 
Clearly, the restriction to planar polygons $\go p$ makes sense in the context of continuous flexible polyhedra, as this condition has to be 
fulfilled around faces where all vertices have valence four\footnote{Assumed that this part of the continuous flexible polyhedra is not rigid.}. 

Our interest in Kokotsakis belts results from our research on continuous flexible polyhedral surfaces;  
especially those composed of rigid planar quads in the combinatorics of a square grid. 
A very well known class of these flexible planar-quad (PQ) surfaces are {\it V-hedra}, which are the discrete analogs of 
Voss surfaces\footnote{Surfaces on which geodesic lines form a conjugate curve network \cite{voss}.} according to \cite{schief}. 
They can easily be characterized by the fact that in each vertex the angles of opposite planar quads are equal.  
The question whether V-hedra can be generalized by dropping the planarity condition of the quads (cf.\ Fig.\ \ref{fig2}) motivated us for the study at hand, which is structured as follows:
We proceed in Section \ref{Review} with a literature review on flexible Kokotsakis belts, where we place emphasis on the so-called
isogonal type\footnote{This notation is in accordance with \cite{stachel}.}, which means that in every polygon vertex both pairs of opposite angles are (1) equal or (2) supplementary.
In Section \ref{sec:sphere} we discuss the spherical image of Kokotsakis belts from a kinematical point of view. 
Based on these consideration we study generalized flexible Kokotsakis belts of the isogonal type in Section \ref{sec:belt}. 
In Section \ref{sec:sqs} we discuss continuous flexible skew-quad (SQ) surfaces, where we focus on V-hedra composed of skew quads in more detail. 
The paper is concluded in Section \ref{sec:conclusion}.

\begin{figure}[t]
\begin{overpic}
    [height=45mm]{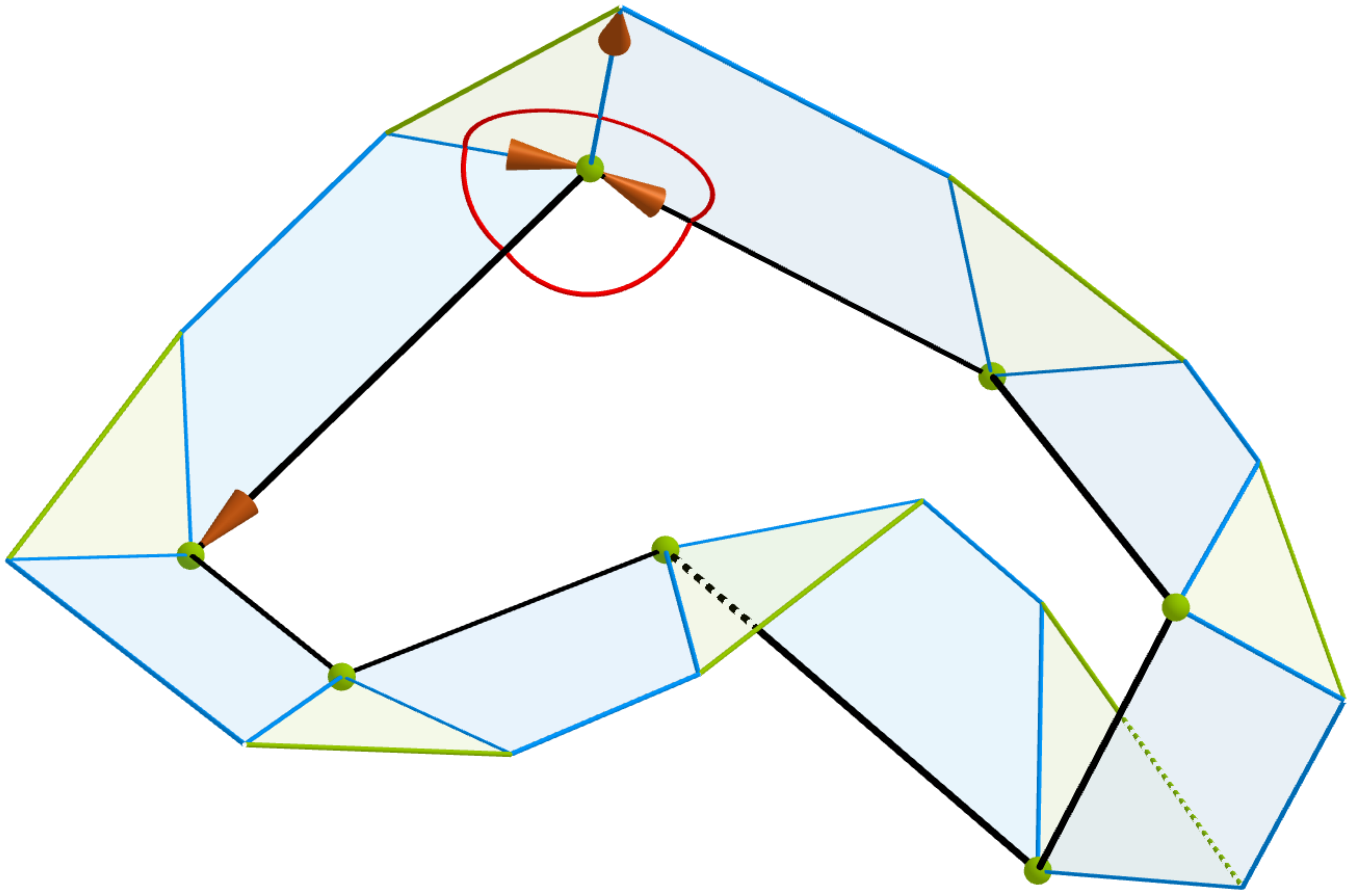}
\begin{scriptsize}
\put(0,0){a)}
\put(42.5,49.5){$V_i$}
\put(42.5,41){$\lambda_i^*$}
\put(48.5,57.5){$\gamma_i^*$}
\put(38.9,60){$\mu_i^*$}
\put(30.2,51){$\delta_i^*$}
\put(72.5,0){$V_0$}
\put(82,20){$V_1$}
\put(45,27.5){$V_{n-1}$}
\end{scriptsize}     
  \end{overpic} 
\hfill
\begin{overpic}
    [height=45mm]{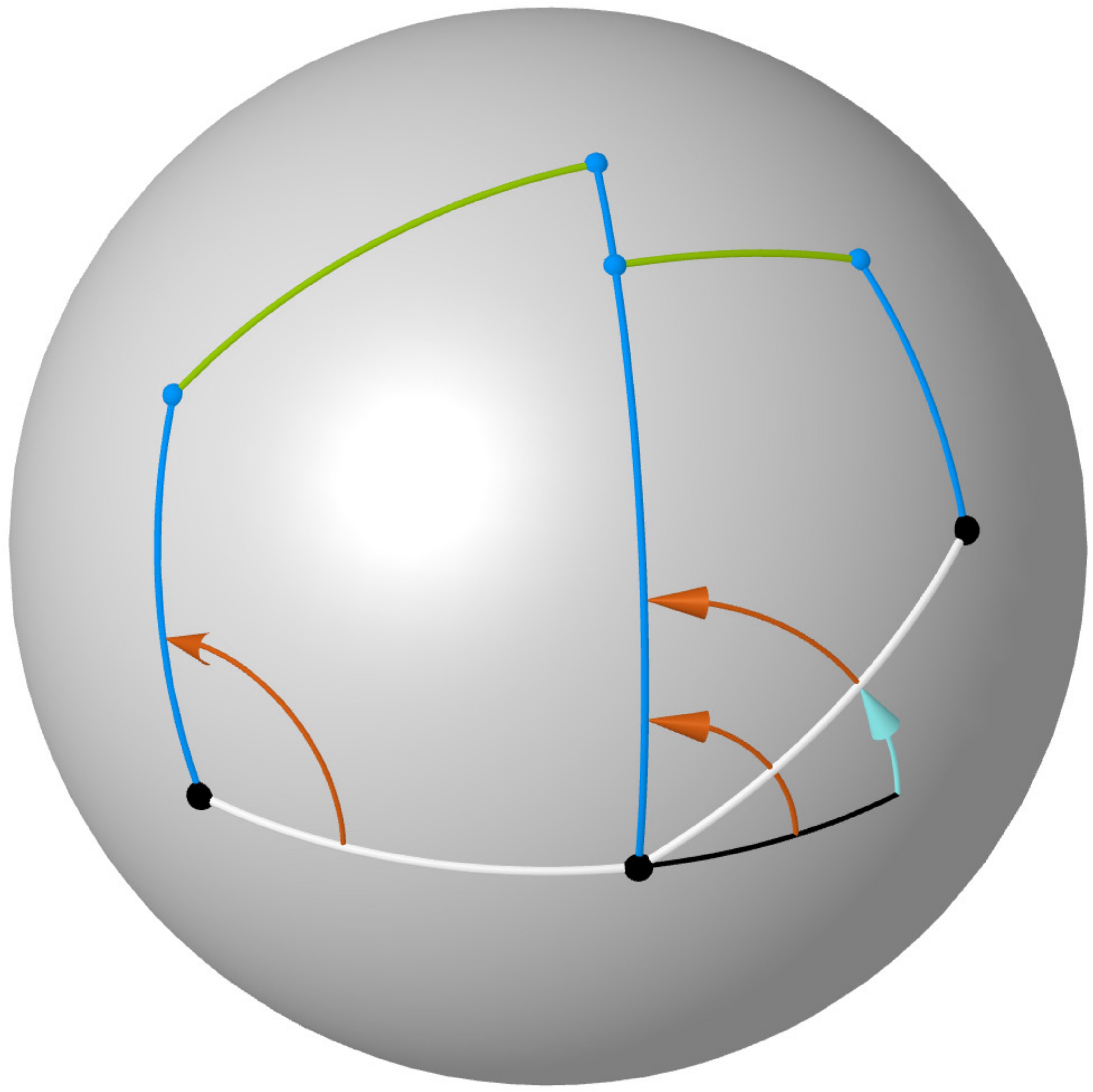}
\begin{scriptsize}
\put(0,0){b)}
\put(9,63){$A_i$}
\put(53,87){$B_i$}
\put(58.5,71.5){$A_{i+1}$}
\put(81,75){$B_{i+1}$}
\put(13,22){$C_{i}$}
\put(52,15){$C_{i+1}$}
\put(77,53){$C_{i+2}$}
\put(10,45){$\gamma_i$}
\put(34,16){$\lambda_i$}
\put(60,55){$\gamma_{i+1}$}
\put(53,55){$\delta_{i}$}
\put(85.5,64.5){$\delta_{i+1}$}
\put(29,79){$\mu_{i}$}
\put(64,79){$\mu_{i+1}$}
\put(82.5,30){$\tau_{i+1}$}
\put(20,41){$\alpha_{i}$}
\put(66,35){$\beta_{i}$}
\put(66,45.5){$\alpha_{i+1}$}
\put(84,39){$\lambda_{i+1}$}
\end{scriptsize}     
  \end{overpic} 
	
\caption{Original Kokotsakis belt: (a) A rigid closed polygonal line $\go p$ (not necessarily planar) with $n$ vertices $V_0,\ldots ,V_{n-1}$ is surrounded by a belt of planar polygons in a way that 
each vertex $V_i$ of $\go p$ has valence four. Moreover, the planar angles 
$\delta_i^*, \gamma_i^*,  \lambda_i^*, \mu_i^* \in(0;\pi)$ are illustrated as well as the orientation of the enclosing line-segments used for the construction of the 
spherical image, which is illustrated in parts in (b). Note that the edge $V_{i-1}V_i$ is mapped to the point $C_i$. 
}
  \label{fig1}
\end{figure}

\begin{figure}[t]
\begin{overpic}
    [height=45mm]{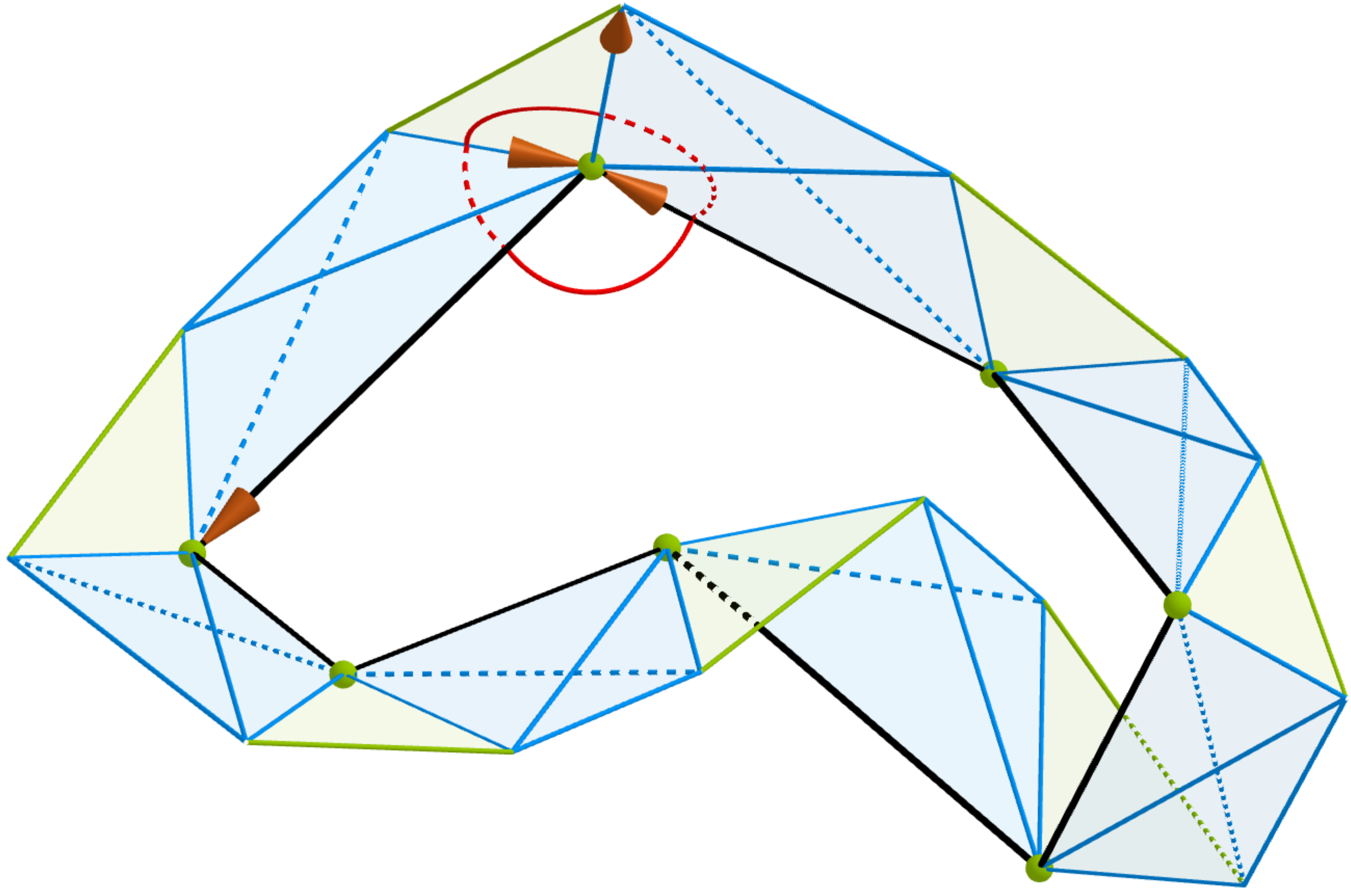}
\begin{scriptsize}
\put(0,0){a)}
\put(42.5,49.5){$V_i$}
\put(42.5,41){$\lambda_i^*$}
\put(48.5,57.5){$\gamma_i^*$}
\put(38.9,60){$\mu_i^*$}
\put(30.2,51){$\delta_i^*$}
\put(72.5,0){$V_0$}
\put(82,20){$V_1$}
\put(45,27.5){$V_{n-1}$}
\end{scriptsize}     
  \end{overpic} 
\hfill
\begin{overpic}
    [height=45mm]{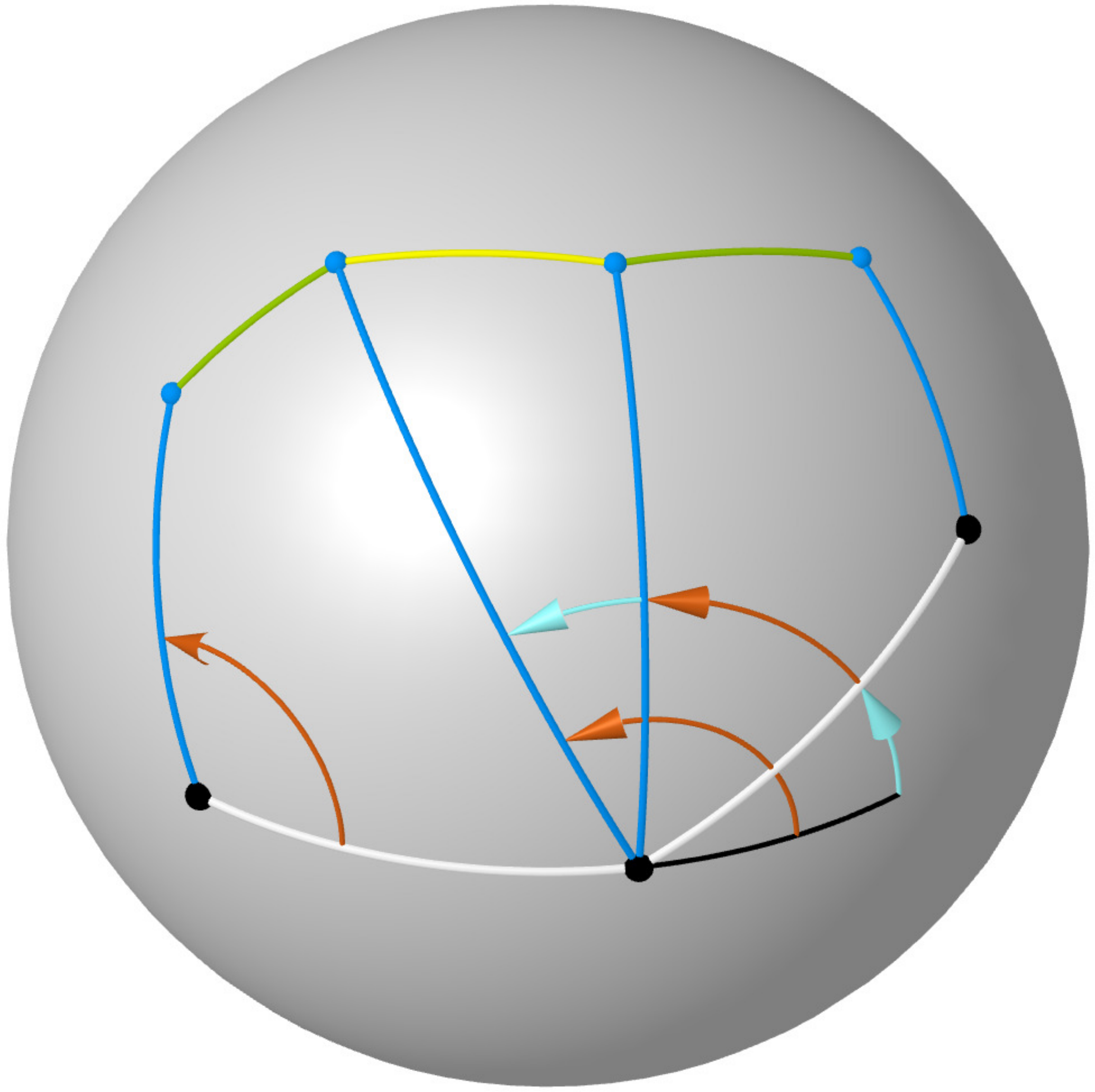}
\begin{scriptsize}
\put(0,0){b)}
\put(9,63){$A_i$}
\put(29,78){$B_i$}
\put(58.5,71.5){$A_{i+1}$}
\put(81,75){$B_{i+1}$}
\put(13,22){$C_{i}$}
\put(52,15){$C_{i+1}$}
\put(77,53){$C_{i+2}$}
\put(10,45){$\gamma_i$}
\put(34,16){$\lambda_i$}
\put(60,55){$\gamma_{i+1}$}
\put(35,50){$\delta_{i}$}
\put(85.5,64.5){$\delta_{i+1}$}
\put(18,71.5){$\mu_{i}$}
\put(64,79){$\mu_{i+1}$}
\put(82.5,30){$\tau_{i+1}$}
\put(20,41){$\alpha_{i}$}
\put(66,35){$\beta_{i}$}
\put(66,45.5){$\alpha_{i+1}$}
\put(84,39){$\lambda_{i+1}$}
\put(48,47){$\zeta_{i+1}$}
\end{scriptsize}     
  \end{overpic} 
	
\caption{Generalized Kokotsakis belt: 
(a) It is obtained from the original Kokotsakis belt illustrated in Fig.\ \ref{fig1} by dropping the planarity condition of the faces adjacent to the polygon line-segments of $\go p$. The resulting 
spatial structures are illustrated as tetrahedra. 
A part of the corresponding spherical image of the generalized Kokotsakis belt is visualized in (b).
}
  \label{fig2}
\end{figure}
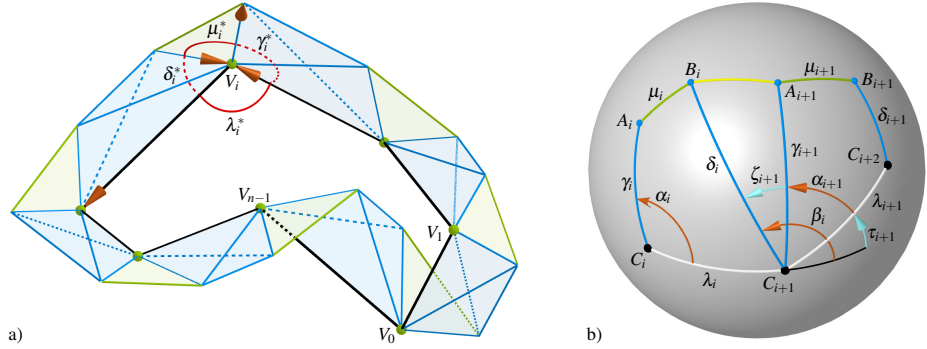

%%%%%%%%%%%%%%%%%%%%%%%%%%%%%%%%%%%%%%%%%%%%%%%%%%%%%%%%%%%%%%%%%%%%%%%%%%%%%%%%%%%%%%%%%%%%%%%%

\subsection{Review on continuous flexible Kokotsakis belts}\label{Review}

Until now only examples of continuous flexible Kokotsakis belts are known where the rigid polygon line $\go p$ is planar as well as
all faces adjacent to its line-segments. Therefore these assumptions hold for the complete review section, which is 
structured along the number $n$ of vertices $V_0,\ldots ,V_{n-1}$ of $\go p$ (cf.\ Fig.\ \ref{fig1}a). 

General results  were only obtained by Kokotsakis \cite{kokotsakis} for the isogonal type, to which for example 
rigidly foldable origami twists \cite{evans} belong as a special case.
For $n=3$ and $n=4$ more results are known, which can be summarized as follows:

\begin{enumerate}[$\bullet$]
\item
{\bf Case $n=3$:}
This case implies continuous flexible octahedra, which are very well studied objects dating back to Bricard \cite{bricard}. 
Especially, the Bricard octahedra of the 3rd type (cf.\ \cite{stachel_t3}) correspond to the isogonal case, which was already pointed out 
by Kokotsakis \cite{kokotsakis}. Moreover, the study of these Kokotsakis belts allows also the determination of 
continuous flexible octahedra with vertices at infinity \cite{naw_inf}. 
\item
{\bf Case $n=4$:} These Kokotsakis belts, which are also known as $(3\times 3)$ complexes, are the building blocks of continuous flexible PQ surfaces according to \cite[Theorem 3.2]{schief}. 
Based on spherical kinematic geometry \cite{stachel}, a partial classification of continuous flexible $(3\times 3)$ building blocks was obtained by Stachel and the author \cite{NS10,Naw11,Naw12}. 
Inspired by this approach, Izmestiev \cite{Izm17} obtained a full classification containing more than 20 cases. 

Note that the first classes of  continuous flexible PQ surfaces were given by Sauer and Graf \cite{graf}; 
namely the so-called T-hedra (see also \cite{sauer,kiumars}) and the already mentioned V-hedra 
(see also \cite{sauer,montagne}).

A rigid-foldable PQ surface which can be developed is a special case of origami. Under the additional condition of flat-foldability (as in case of the popular Miura-ori)
Tachi \cite{Tac09,Tac10} developed computational tools to design surfaces, where each vertex is of the isogonal type. 
Recently, Feng et al.\ \cite{feng} gave a complete analysis of the flat-foldable case, which can also be used for design tasks \cite{feng2022}. 

Within the field of computational design Jiang et al.\ \cite{jiang} presented recently an optimization technique to penalize  
an isometrically deformed surface with planar quads. Its design space is restricted to rigid-foldable quad-surfaces 
which can be seen as a discretization of flexible smooth surfaces (e.g.\ Voss surfaces, profile-affine surfaces \cite{graf,sauer}). 
\end{enumerate}

%%%%%%%%%%%%%%%%%%%%%%%%%%%%%%%%%%%%%%%%%%%%%%%%%%%%%%%%%%%%%%%%%%%%%%%%%%%%%%%%%%%%%%%%%%%%%%%%
%%%%%%%%%%%%%%%%%%%%%%%%%%%%%%%%%%%%%%%%%%%%%%%%%%%%%%%%%%%%%%%%%%%%%%%%%%%%%%%%%%%%%%%%%%%%%%%%
%%%%%%%%%%%%%%%%%%%%%%%%%%%%%%%%%%%%%%%%%%%%%%%%%%%%%%%%%%%%%%%%%%%%%%%%%%%%%%%%%%%%%%%%%%%%%%%%

\section{Spherical image of Kokotsakis belts}\label{sec:sphere}

In order to get a consistent notation for the construction of the spherical image of the  Kokotsakis belt 
we orient the line-segments meeting at a vertex $V_i$ according to Figs.\ \ref{fig1}a and \ref{fig2}a, respectively. 
Taking this orientation of the line-segments into account, the spherical 4-bar mechanism, which corresponds 
with the arrangement of faces around $V_i$, has the following spherical bar lengths:
\begin{equation}
\delta_i=\pi-\delta_i^*, \quad \gamma_i=\pi-\gamma_i^*,  \quad \lambda_i=\pi-\lambda_i^*, \quad \mu_i=\pi-\mu_i^*
\end{equation}
for the index\footnote{Note that in the remainder of the paper the indices are taken modulo $n$.} $i=0,\ldots ,n-1$. 
The spherical image of faces around two adjacent vertices $V_i$ and $V_{i+1}$ is illustrated in Figs.\ \ref{fig1}b and \ref{fig2}b, 
which show the motion transmission from the vertex $C_{i}$ over $C_{i+1}$ to $C_{i+2}$ by two coupled spherical 4-bar mechanisms. 
Note that in the isogonal case these 4-bar mechanisms are so-called spherical isograms fulfilling one of the following two conditions:
\begin{equation}\label{eq:types}
(1)\phm \lambda_i=\mu_i, \phm \delta_i=\gamma_i, \qquad\qquad (2) \phm \lambda_i+\mu_i=\pi, \phm \delta_i+\gamma_i=\pi. 
\end{equation}
Note that these two types are related by the replacement of one of the vertices of the spherical isogram by its 
antipodal point, which does not change its motion. In Section \ref{sec:belt} we show that we can restrict to type (1) without loss of 
generality by assuming an appropriate choice of orientations.
In the following we use the half-angle substitutions
\begin{equation}
\sin{\alpha_i}=\tfrac{2a_i}{1+a_i^2}, \quad \cos{\alpha_i}=\tfrac{1-a_i^2}{1+a_i^2}, \quad
\sin{\beta_i}=\tfrac{2b_i}{1+b_i^2}, \quad \cos{\beta_i}=\tfrac{1-b_i^2}{1+b_i^2}, 
\end{equation}
in order to end up with algebraic expressions. 
It is well known (e.g.\ \cite{stachel}) that the input angle $\alpha_i$ and the output angle 
$\beta_i$ of the $i$-th spherical isogram of type (1) of Eq.\ (\ref{eq:types}) are related by
\begin{equation}\label{eq:b}
b_i=f_ia_i \quad \text{with} \quad f_i\neq 0 \quad \text{and} \quad
\quad f_i=\tfrac{\sin{\delta_i}\pm \sin{\lambda_i}}{\sin{(\delta_i-\lambda_i)}}.
\end{equation}
The two options in the expression for $f_i$ implied by the $\pm$ sign refer to the case whether the 
motion transmission corresponds to that of a  
spherical parallelogram ($\Leftrightarrow$ $f_i>0$) or
spherical antiparallelgram ($\Leftrightarrow$ $f_i<0$), respectively. 
Note that the degenerated cases ($\delta_i=\lambda_i$ and $\delta_i+\lambda_i=\pi$) of the spherical isogram are excluded by the condition $f_i\neq 0$ 
given in Eq.\ (\ref{eq:b}).

The angles $\beta_i$ and $\alpha_{i+1}$ are related over the offset angle $\varepsilon_{i+1}$; 
i.e.\ $\beta_i+\varepsilon_{i+1}=\alpha_{i+1}$. This means that $\varepsilon_{i+1}$ gives only the shift 
between the output angle $\beta_i$ of the $i$-th isogram to the input angle $\alpha_{i+1}$ of the 
$(i+1)$-th isogram. This yields the relation:
\begin{equation}
\tan{\alpha_{i+1}}=\tfrac{\tan{\beta_i}+\tan{\varepsilon_{i+1}}}{1-\tan{\beta_i}\tan{\varepsilon_{i+1}}}.
\end{equation} 
Using the half-angles and the Weierstrass substitution $e_{i+1}:=\tan{\tfrac{\varepsilon_{i+1}}{2}}$
yield
\begin{equation}\label{eq:a}
a_{i+1}=\tfrac{b_i+e_{i+1}}{1-b_ie_{i+1}}.
\end{equation}

Note that the spherical arcs $B_iC_{i,i+1}$ and $A_{i+1}C_{i,i+1}$ enclose the angle  
$\zeta_{i+1}:=\varepsilon_{i+1}+\tau_{i+1}$ (cf.\ Fig.\ \ref{fig2}b), where the latter angle is the torsion angle of the spatial 
polygon $\go p$, which is defined as the angle enclosed by the spherical arcs $C_{i}C_{i+1}$ 
and $C_{i+1}C_{i+2}$. 
From the polygon $\go p$ the angles $\tau_{i+1}$ can be computed as the angle of rotation about the oriented axis $V_iV_{i+1}$, which brings the 
plane $[V_{i-1},V_i,V_{i+1}]$ to the plane $[V_{i},V_{i+1},V_{i+2}]$. Therefore $\tau_{i+1}$, which is within the interval $(-\pi;\pi]$, 
can be computed as:
\begin{equation}
\tau_{i+1}=\sign{(o)}\arccos{\left(\tfrac{(c_i\kreuz c_{i+1})\,(c_{i+1}\kreuz c_{i+2})}{\|c_i\kreuz c_{i+1}\| \|c_{i+1}\kreuz c_{i+2}\|}\right)}\quad
\text{with}\quad o:=(c_i\kreuz c_{i+1})\,c_{i+2}
\end{equation}
where $c_i$ denotes the vector from $V_{i-1}$ to $V_i$.

\begin{remark}\label{rem1}
For the original Kokotsakis belt (cf.\ Fig.\ \ref{fig1}) the angle $\zeta_{i+1}$ is zero ($\Rightarrow$ $\varepsilon_{i+1}=-\tau_{i+1}$) 
or $\pi$ ($\Rightarrow$ $\varepsilon_{i+1}=\pi-\tau_{i+1}$) for all $i=0,\ldots ,n-1$.  
Note that $\go p$ is a planar curve if all $\tau_{i+1}$ are either zero or $\pi$. \hfill $\diamond$
\end{remark}

%%%%%%%%%%%%%%%%%%%%%%%%%%%%%%%%%%%%%%%%%%%%%%%%%%%%%%%%%%%%%%%%%%%%%%%%%%%%%%%%%%%%%%%%%%%%%%%%
%%%%%%%%%%%%%%%%%%%%%%%%%%%%%%%%%%%%%%%%%%%%%%%%%%%%%%%%%%%%%%%%%%%%%%%%%%%%%%%%%%%%%%%%%%%%%%%%
%%%%%%%%%%%%%%%%%%%%%%%%%%%%%%%%%%%%%%%%%%%%%%%%%%%%%%%%%%%%%%%%%%%%%%%%%%%%%%%%%%%%%%%%%%%%%%%%

\section{Continuous flexible Kokotsakis belts of the isogonal type}\label{sec:belt}

According to \cite[Theorem 1]{stachel} the Kokotsakis belt is continuous flexible if and only if the spherical image has this property. 
Now we will show that any Kokotsakis belt of the isoganal type can be identified with 
a spherical mechanism, which is only composed of spherical isograms of type (1) in  Eq.\ (\ref{eq:types}):
 
We start with the spherical image of $\go p$, i.e.\ the points $C_0,\ldots, C_{n-1}$ 
and construct the spherical points $A_0$ and $B_0$ according to Section \ref{sec:sphere}. 
If the spherical isogram $C_0C_1B_0A_0$ is of type (2) then we replace $B_0$ by its antipode. 
Then we proceed as follows around the spherical image of the polyline $\go p$; i.e. for $i=0,\ldots,n-1$:
\begin{enumerate}[a.]
\item
In the case where the two antipodal points, which are candidates for $A_{i+1}$, correspond with the values zero and $\pi$ for $\varepsilon_{i+1}$, 
we have to choose the one which implies  $\varepsilon_{i+1}=0$ as $\varepsilon_{i+1}=\pi$ is not covered by Eq.\ (\ref{eq:a}). 
In any other case  $A_{i+1}$ can be chosen arbitrary from the corresponding set of two antipodal points.  
\item
$B_{i+1}$ has to be chosen from the corresponding set of two antipodal points such that the 
spherical isogram  $C_{i+1}C_{i+2}B_{i+1}A_{i+1}$ is of type (1).  
\end{enumerate} 
We can end up in two situations; either $A_n=A_0$ and we are done or $A_n$ is the antipodal point of $A_0$. 
In the latter case we denote by $j$ the highest possible index within the set $\left\{0,\ldots,n-1\right\}$ for which the  
choice of $A_{i+1}$ was done arbitrarily in step (a). Then we replace all $A_{i+1}$ and $B_{i+1}$ with $i\geq j$ by their antipodal points which yields $A_n=A_0$.  

Note that such a $j$ has to exist as otherwise we can construct the following contradiction: 
No $j$ exists if and only if there are no shifts; i.e.\  
$e_{0}=e_{1}= \ldots  =e_{n-1}=0$.  
As a consequence $\alpha_0=\beta_0=0$ implies $\alpha_{i+1}=\beta_{i+1}=0$ for all $i\in\left\{0,\ldots,n-1\right\}$, 
which already shows that in this case $A_n=A_0$ has to \bigskip hold.

As a consequence of the above considerations one can write down the condition for continuous flexibility 
of any Kokotsakis belt of the isogonal type, where the rigid polygon $\go p$ has $n>2$ vertices, as 
\begin{equation}
a_0-a_{n}=0.
\end{equation}
In this so-called closure condition we substitute $a_{n}$ by 
\begin{equation}
a_{i}=\tfrac{a_{i-1}f_{i-1}+e_{i}}{1-a_{i-1} f_{i-1} e_{i}} 
\end{equation}
which results from Eq.\ (\ref{eq:a}) under consideration of Eq.\ (\ref{eq:b}). 
By iterating this substitution (in total $n$ times) we end up with an expression of the form 
$q_2a_0^2 + q_1a_0 + q_0=0$
where $q_2,q_1,q_0$ are functions in $f_0,\ldots , f_{n-1}, e_{0}, \ldots , e_{n-1}$. 
This means that the spherical coupler arms $A_0C_0$ and $A_{n}C_{0}$ 
coincide for all input angles $\alpha_0$ if and only if the following necessary and sufficient conditions for 
continuous mobility are fulfilled:
\begin{equation}\label{eq:con}
q_2=0, \quad q_1=0, \quad q_0=0.
\end{equation}
This results in the following theorem:

\begin{theorem}
For a given polyhedral curve $\go p$ with $n$ vertices,  
there exists at least a $(2n-3)$-dimensional set of continuous flexible Kokotsakis belts of the isogonal type over $\CC$. 
\end{theorem}

By taking a closer look at $q_2=0$ it can easily be seen that the terms linear in $e_i$ are given by
$f_0e_1$ and $f_0 \ldots f_{k-1}e_k$ for $k=2,\ldots, n$. 
In the equation  $q_0=0$ the linear terms in  $e_i$ are $e_0$, $e_{n-1}f_{n-1}$ and $e_kf_k\ldots f_{n-1}$ for $k=1,\ldots, n-2$.
Therefore each of the two conditions $q_2=0$ and $q_0=0$ can only be fulfilled 
independently from the choice of the $f_i$'s 
if there are no shifts; i.e.\ $e_{0}=e_{1}= \ldots  =e_{n-1}=0$.
Note that these are not only necessary conditions but already sufficient ones as they imply $q_2=q_0=0$. 
In this case the remaining condition $q_1=0$ simplifies to $f_0f_1\cdots f_{n-1}=1$
and we end up with a $(n-1)$-dimensional set of continuous flexible Kokotsakis belts of the isogonal type over $\CC$. 
Note that  
$e_{0}=e_{1}= \ldots  =e_{n-1}=0$ 
only implies planarity of $\go p$ if we assume the faces to be \bigskip planar.

For a spatial polyline $\go p$ with planar faces  the spherical coupler arms $B_iC_{i+1}$ and $A_{i+1}C_{i+1}$ are aligned. 
Therefore all $e_{i+1}$ are determined (cf.\ Remark \ref{rem1}) and we get:

\begin{theorem}
For a given polyhedral curve $\go p$ with $n>3$ vertices, there exists at least a
$(n-3)$-dimensional set of continuous flexible Kokotsakis belts with planar faces of the isogonal type over $\CC$. 
For planar curves $\go p$ (which is always true for $n=3$) this dimension raises to $(n-1)$. 
\end{theorem}

\begin{figure}[t]
\begin{overpic}
    [height=45mm]{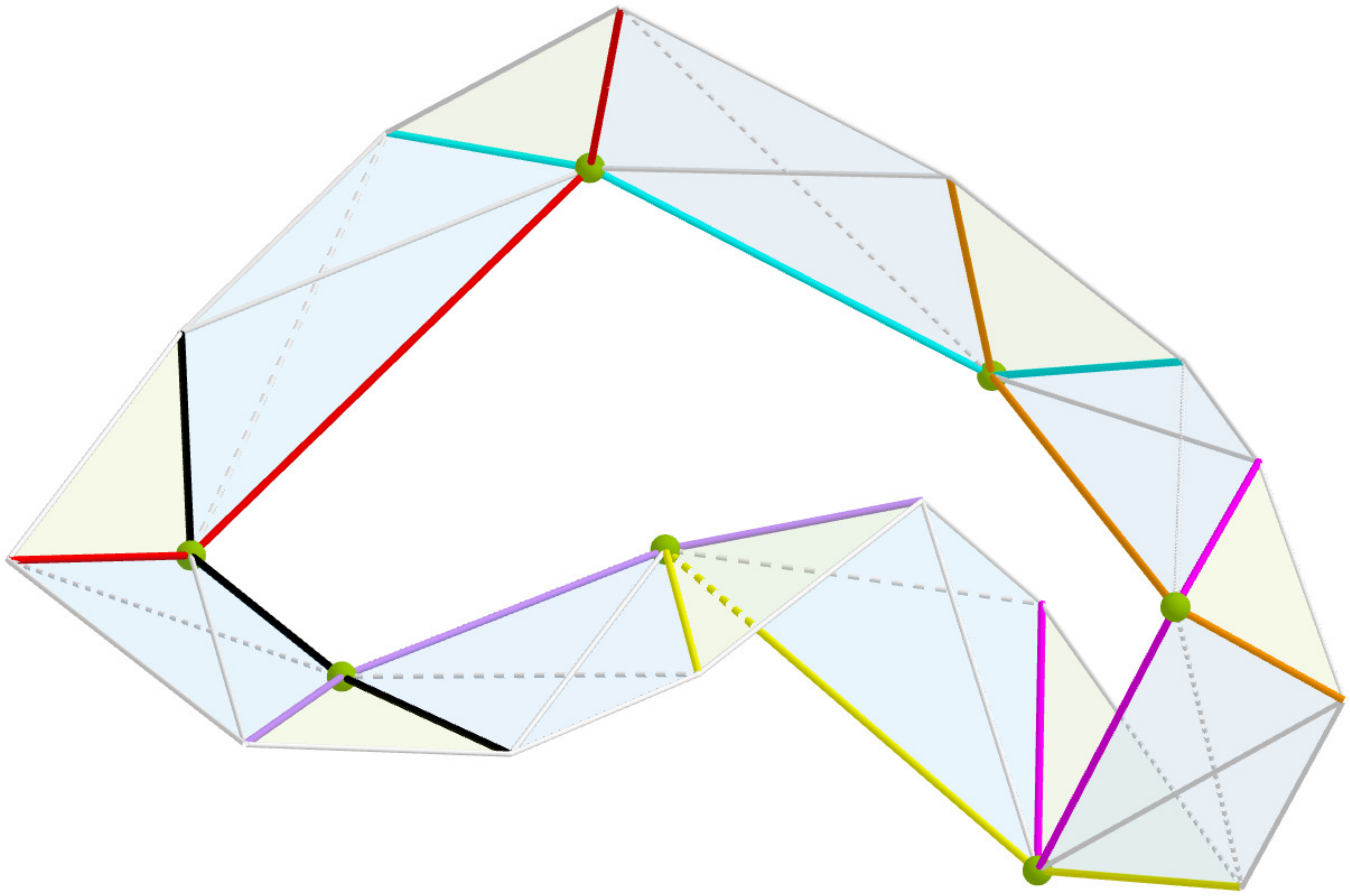}
\begin{scriptsize}
\put(0,0){a)}
\put(42.5,49.5){$V_i$}
\put(72.5,0){$V_0$}
\put(82,20){$V_1$}
\put(45,27.5){$V_{n-1}$}
\end{scriptsize}     
  \end{overpic} 
\hfill
\begin{overpic}
    [height=35mm]{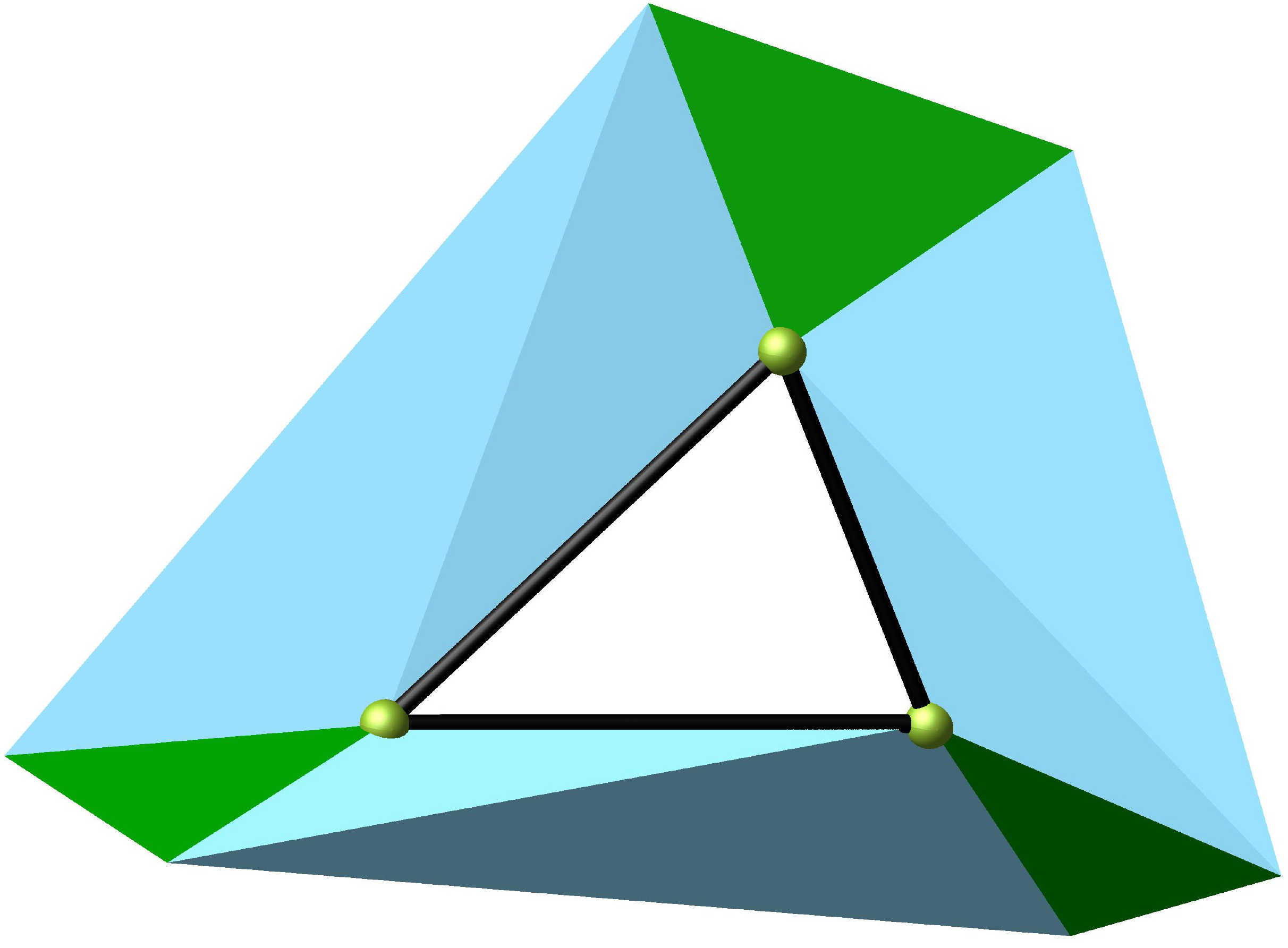}
\begin{scriptsize}
\put(0,0){b)}
\end{scriptsize}     
  \end{overpic} 
	
\caption{Generalized Kokotsakis belt of the isogonal type: 
(a) Edges with the same absolute values of their rotation angles during the  continuous flexibility are illustrated with the same color. 
(b) For $n=3$ we obtain an overconstrained angle-symmetric 6R linkage. 
}
  \label{fig3}
\end{figure}
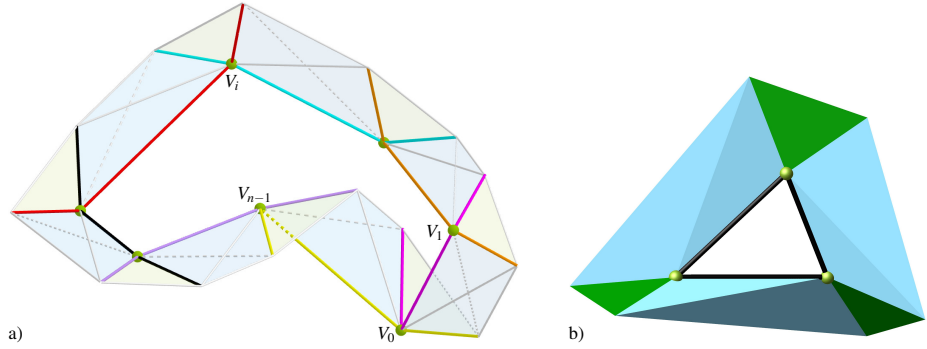

\subsection{Property regarding the rotation angles}\label{sec:prop}

According to  \cite[\S 8]{kokotsakis} opposite angles in a spherical isogram are either equal or complete each other to $2\pi$. As a consequence 
opposite dihedral angles along edges meeting in a vertex $\go V_i$ have at each time instant $t$ the same absolute value of their angular velocities. 
Therefore the absolute values of the rotation angles around these two edges are the same (measured from an initial starting configuration). 
As one of the dihedral angels is the angle about an edge $\go V_{i}\go V_{i+1}$ of the polygon $\go p$, this property holds for the 
two spherical 4-bars, which have the common point $\go C_{i+1}$ (cf.\ Figs.\ \ref{fig1}b and \ref{fig2}b, respectively). 
Therefore the same absolute values of the rotation angle can always be assigned to 
three edges within a continuous flexible Kokotsakis belt of the isogonal type (cf.\ Fig.\ \ref{fig3}a).

\begin{example}
Now we consider the case $n=3$. For any choice of $\delta_i$ and $\gamma_i$ for $i=1,2,3$ and $\gamma_1+\gamma_2+\gamma_3=2\pi$ (closure condition of 
central triangle) there exist $e_0,e_1,e_2\in\CC$ such that we get a continuous flexible Kokotsakis belt of the isogonal type. 
The resulting structure can be seen as an overconstrained 6R loop (cf.\ Fig.\ \ref{fig3}b), which  belongs to the third class of so-called angle-symmetric 6R linkages \cite{zijia} due to the 
above discussed angle property. 
Note that for $e_0=e_1=e_2=0$ we get the already mentioned Bricard octahedron of type III (cf.\ case $n=3$ in Section \ref{Review}). 
\hfill $\diamond$
\end{example}

%%%%%%%%%%%%%%%%%%%%%%%%%%%%%%%%%%%%%%%%%%%%%%%%%%%%%%%%%%%%%%%%%%%%%%%%%%%%%%%%%%%%%%%%%%%%%%%%
%%%%%%%%%%%%%%%%%%%%%%%%%%%%%%%%%%%%%%%%%%%%%%%%%%%%%%%%%%%%%%%%%%%%%%%%%%%%%%%%%%%%%%%%%%%%%%%%
%%%%%%%%%%%%%%%%%%%%%%%%%%%%%%%%%%%%%%%%%%%%%%%%%%%%%%%%%%%%%%%%%%%%%%%%%%%%%%%%%%%%%%%%%%%%%%%%

\section{Continuous flexible SQ surfaces}\label{sec:sqs}

On page 168 of Sauer's book \cite{sauer} the following open problem is mentioned: {\it Do there exist continuous flexible SQ surfaces?} 
We answer this question positively by constructing $(3\times 3)$ building blocks of  continuous flexible V-hedra  with skew quads within this section. 
The restriction to these substructures is sufficient as Theorem 3.2 of \cite{schief} can be generalized in the following way:

\begin{theorem}\label{th:schief}
A non-degenerate SQ surface is continuous flexible, if and only if this holds true for every $(3\times 3)$ building block.
\end{theorem}

\begin{proof}
The arguments used for the proof of Theorem 3.2 of \cite{schief} do not rely on the planarity of the involved quads. \hfill $\BewEnde$
\end{proof}

\subsection{Associated overconstrained mechanism}\label{subsec:associate}

We start this section with the definition of reciprocal-parallel quad meshes:

\begin{definition}
Two quad meshes $\mathcal{Q}$ and $\mathcal{V}$ are called reciprocal-parallel if the following conditions are fulfilled: 
\begin{enumerate}[$\star$]
\item
$\mathcal{Q}$ and $\mathcal{V}$ are combinatorial dual; i.e.\ vertices of one mesh correspond to the faces of the other and vice versa.
\item
The edges of both meshes are related by the implied bijection that edges of adjacent face are mapped to edges between corresponding adjacent vertices and vice versa.
\item
Edges, which are related by this bijection, are parallel.
\end{enumerate}
\end{definition}

Sauer \cite{sauer} showed that every infinitesimal flexible quad surface $\mathcal{Q}$ possesses in general a unique (up to scaling) 
reciprocal-parallel quad mesh $\mathcal{V}$ with rigid vertices.  
The reciprocal-parallel surface of the latter mesh $\mathcal{V}$ is only uniquely determined (up to scaling) if $\mathcal{Q}$ is composed of skew quads, otherwise 
there exist infinitely many, which are in a parallelism relation\footnote{Corresponding faces and edges of 
these meshes are parallel.} to each other (cf.\ \cite[Theorem 16.22]{sauer}). 

The corresponding deformation of the $\mathcal{V}$ mesh during the continuous flexion of $\mathcal{Q}$ has to be a conformal transformation, as the 
vertex stars are rigid. The corresponding kinematic structure of $\mathcal{V}$ is composed of rigid vertex stars linked by cylindrical joints (cf.\  Fig.\ \ref{fig5}). Note that in general such a structure only has the trivial mobility resulting from the homothetic transformation. 
This motion can be omitted by fixing the length of one edge in the structure.
These modified linkages are in general rigid but those stemming from continuous flexible quad surfaces $\mathcal{Q}$ have 
an overconstrained motion.

%%%%%%%%%%%%%%%%%%%%%%%%%%%%%%%%%%%%%%%%%%%%%%%%%%%%%%%%%%%%%%%%%%%%%%%%%%%%%%%%%%%%%%%%

\subsection{V-hedra with skew quads}\label{sec:vhedra}

We study this class in more detail, as it is a generalization of V-hedra with planar quads having many applications to structural engineering practice 
\cite{montagne,mitchell}. 

For $n=4$ the equations $q_2=q_1=q_0=0$ of Eq.\ (\ref{eq:con}) read as follows: 
\begin{equation}\label{eq:qs}
\begin{split}
q_2 := &f_0[e_0f_3(e_1e_2f_2 + e_2e_3f_1 + e_1e_3 - f_1f_2) + e_1e_2e_3f_2 - e_3f_1f_2 - e_2f_1 - e_1], \\
q_1 := &e_1e_2f_0f_2f_3 +e_2e_3f_0f_1f_3  + e_1e_3f_0f_3 - e_1e_3f_1f_2 - f_0f_1f_2f_3  - e_1e_2f_1 \\
&- e_2e_3f_2 + 1 +e_0(e_1e_2e_3f_1f_3 - e_1e_2e_3f_0f_2 - e_1f_1f_2f_3 + e_3f_0f_1f_2 \\
&+ e_2f_0f_1 - e_2f_2f_3 + e_1f_0 - e_3f_3), \\
q_0 :=&e_0(e_1e_3f_1f_2 + e_1e_2f_1 + e_2e_3f_2 -1)
+ f_3(e_1e_2e_3f_1  - e_1f_1f_2 - e_2f_2 - e_3).
\end{split}
\end{equation}
We can solve this set of equations explicitly for $e_1,e_2,e_3$ in dependence of $e_0,f_0,\ldots,f_3$, which yields the following two solutions: 
\begin{equation}\label{sol_es}
\begin{split}
e_1&=\tfrac{
e_0f_0f_2(f_1^2-1)(f_3^2-1)\pm R_1R_2
} 
{e_0^2(f_0 f_2 - f_1 f_3) (f_0-f_1 f_2 f_3 ) + (f_0 f_3 - f_1 f_2) (f_0 f_2 f_3 - f_1)}, \\
e_2&=\tfrac{
\mp R_1R_2
}
{e_0^2(f_0 f_1 - f_2 f_3) (f_0 f_2 - f_1 f_3) + (f_0 f_1 f_3 - f_2) (f_0 f_2 f_3 - f_1)
}, \\
e_3&=\tfrac{
e_0f_1f_3(f_0^2 - 1)(f_2^2 - 1)\pm R_1R_2}
{e_0^2(f_0 f_2 - f_1 f_3) (f_0 f_1 f_2 - f_3)+(f_0 f_3 - f_1 f_2) (f_0 f_1 f_3 - f_2)
}
\end{split}
\end{equation}
with
\begin{equation}
\begin{split}
R_1:&=\sqrt{e_0^2(f_0f_1 - f_2f_3)(f_0f_2 - f_1f_3) + (f_0f_1f_3 - f_2)(f_0f_2f_3 - f_1)}, \\
R_2:&=\sqrt{e_0^2(f_0f_1f_2 - f_3)(f_1f_2f_3 - f_0) + (f_0f_1f_2f_3 - 1)(f_1f_2 - f_0f_3)}.
\end{split}
\end{equation}

\begin{remark}
Alternatively, the above given equations 
$q_2=q_1=q_0=0$ from Eq.\ (\ref{eq:qs}) can also be solved explicitly for 
$f_1,f_2,f_3$ in dependence of $f_0,e_0,\ldots,e_3$, which yield the following two solutions:
\begin{equation}
\begin{split}
f_1&=\tfrac{-(f_0^2+1)e_0e_1(e_3^2 + 1)(e_2^2 - 1) 
+ f_0(e_0^2e_1^2e_2^2 - e_0^2e_1^2e_3^2 - e_0^2e_2^2e_3^2 - e_1^2e_2^2e_3^2 + e_0^2 + e_1^2 + e_2^2 - e_3^2)
\pm R_3R_4
} 
{2e_2(e_3^2 + 1)(e_0f_0 + e_1)(e_0e_1 - f_0)}, \\
f_2&=\tfrac{\phm (f_0^2+1)e_0e_1(e_3^2 + 1)(e_2^2 + 1)
-f_0(e_0^2e_1^2e_2^2 + e_0^2e_1^2e_3^2 - e_0^2e_2^2e_3^2 - e_1^2e_2^2e_3^2 - e_0^2 - e_1^2 + e_2^2 + e_3^2)
\mp R_3R_4}
{2e_2e_3f_0(e_1^2 + 1)(e_0^2 + 1)}, \\
f_3&=\tfrac{-(f_0^2+1)e_0e_1(e_3^2 - 1)(e_2^2 + 1)
-f_0(e_0^2e_1^2e_2^2 - e_0^2e_1^2e_3^2 + e_0^2e_2^2e_3^2 + e_1^2e_2^2e_3^2 - e_0^2 + e_1^2 + e_2^2 - e_3^2)
\pm R_3R_4}
{2e_3(e_2^2 + 1)(e_1f_0 + e_0)(e_0e_1 - f_0)}
\end{split}
\end{equation}
with
\begin{equation}
\begin{split}
R_{3,4}:=&[f_0(e_0^2e_1^2e_2^2 + e_0^2e_1^2e_3^2 - e_0^2e_2^2e_3^2 - e_1^2e_2^2e_3^2 - e_0^2 - e_1^2 + e_2^2 + e_3^2)\\
&-(f_0^2+1) e_0e_1(e_3^2 + 1)(e_2^2 + 1) \pm 2f_0e_2e_3(e_1^2 + 1)(e_0^2 + 1)]^{\tfrac{1}{2}},
\end{split}
\end{equation}
where $R_3$ corresponds to the plus sign and $R_4$ to the minus sign. \phm \hfill $\diamond$
\end{remark}

Note that for a given skew central quad $\go p$ and a set of real values $e_0,\ldots,e_3,f_0,\ldots ,f_3$
fulfilling the three equations $q_2=q_1=q_0=0$, the missing geometric 
parameters $\delta_i$ can be computed from the equation
$\sin{\delta_i}\pm \sin{\lambda_i}-f_i\sin{(\delta_i-\lambda_i)}=0$ (cf.\ Eq.\ (\ref{eq:b})).  
For the minus sign we get 
one further real solution beside the excluded degenerate case $\delta_i=\lambda_i$. 
By shifting these two values obtained for $\delta_i$ by $\pi$ we obtain the solutions of the 
equation with respect to the plus sign. Therefore we get a unique 
value for $\delta_i\in(0;\pi)$ with $\delta_i\neq \lambda_i$ for each $i=0,\ldots, 3$.

%%%%%%%%%%%%%%%%%%%%%%%%%%%%%%%%%%%%%%
\begin{example}
The coordinates of the vertices of the skew central quad $\go p$ are given by:
\begin{equation}
V_0=(5,0,0)^T, \quad
V_1=(4,3,0)^T, \quad
V_2=(1,2,2)^T, \quad
V_3=(0,0,0)^T,
\end{equation}
from which the angles $\lambda_i$ and $\tau_i$ for $i=0,\ldots,3$ can be calculated. 
Moreover, the input data is completed by the values:
\begin{equation}
e_0=100, \quad d_0=0.3, \quad d_1=0.15, \quad d_2=0.2, \quad d_3=0.25,
\end{equation}
where $d_i=\tan{\tfrac{\delta_i}{2}}$.  
From that we can compute the $f_i$ values according to Eq.\ (\ref{eq:b}) with respect 
to the minus sign for all $i=0,\ldots,3$. 
Then the formulas for the solution set related to the upper sign in  Eq.\ (\ref{sol_es}) 
yield
\begin{equation}
e_1=-0.86081001, \quad
e_2=-5.06077939, \quad
e_3= 0.57043281.
\end{equation}
One configuration of the resulting continuous flexible $(3\times 3)$ building block of a 
V-hedra  with skew quads is illustrated in Fig.\ \ref{fig4}, where also the corresponding
spherical image is displayed. The associated overconstrained mechanism implied by the 
reciprocal-parallelism (cf.\ Section \ref{subsec:associate}) is shown in Fig.\ \ref{fig5}. 
In the captions of Figs.\ \ref{fig4} and \ref{fig5} we also provide links to {\tt gif} animations 
showing the overconstrained motion of these three mechanisms. 
\hfill $\diamond$
\end{example}

%%%%%%%%%%%%%%%%%%%%%%%%%%%%%%%%%%%%%%%%%%%%%%%%%%%%%%%%%%

\noindent
We close this section by making the following two final comments:
\begin{enumerate}[$\bullet$]
\item
The edges of the V-hedra can be subdivided into two families of discrete parameter lines, which are 
called $u$-polylines and $v$-polylines for short.
Due to the property pointed out in Section \ref{sec:prop}, the rotation angles along any $u$-polyline or $v$-polyline 
are the same. Note that this property is well known for V-hedra with planar quads 
(cf.\ \cite[page 529]{graf}) but also holds for the skew case. 
\item
In view of Section \ref{subsec:associate} it should be noted that there is a
further remarkable relation to an overconstrained mechanism beside the one illustrated in Fig.\ \ref{fig5}. 
As already pointed out by Sauer \cite{sauer}
the vertex star fulfilling the isogonality condition is reciprocal-parallel to a skew isogram, which has the following additional property: 
If the four bars of the isogram are hinged in the vertices by rotational joints, which are orthogonal to the plane spanned by the linked bars (cf.\ Fig.\ \ref{fig6}), then 
one obtains a so-called Bennett mechanism \cite{bennett}.  
This is the only non-trivial mobile 4R loop. 

If the quads of the V-hedra $\mathcal{Q}$ are skew then the four axes of the Bennett mechanisms, which can be associated with a vertex of the mesh $\mathcal{V}$, differ from each other (cf.\ Fig.\ \ref{fig6}). 
Only in the case where $\mathcal{Q}$ is a V-hedra with 
planar quads, each vertex of $\mathcal{V}$ can uniquely be associated with one rotational axis orthogonal to 
the planar vertex star. Then the resulting network of Bennett mechanisms is highly mobile\footnote{The degree of the mobility corresponds to the number of rows plus columns of $\mathcal{V}$ minus one 
(cf.\ Sauer \cite[Theorem 11.18]{sauer}).}. Finally it should be noted that in this case $\mathcal{V}$ is a discrete pseudospherical surface \cite{schief,sauer,wunderlich}.  
\end{enumerate}

\begin{figure}[t]
\begin{overpic}
    [height=65mm]{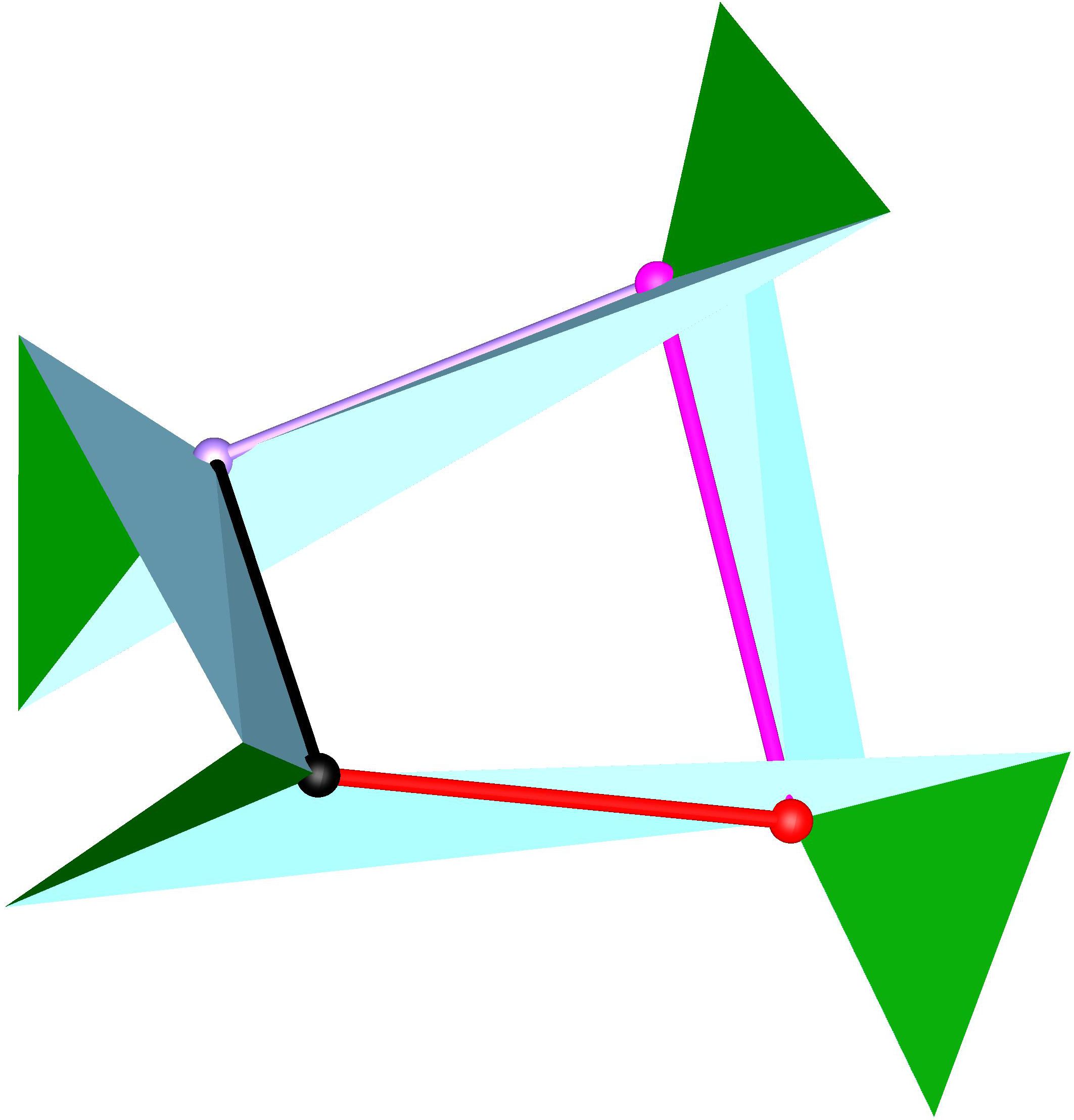}
\begin{scriptsize}
\put(0,0){a)}
\end{scriptsize}     
  \end{overpic} 
\hfill
\begin{overpic}
    [height=45mm]{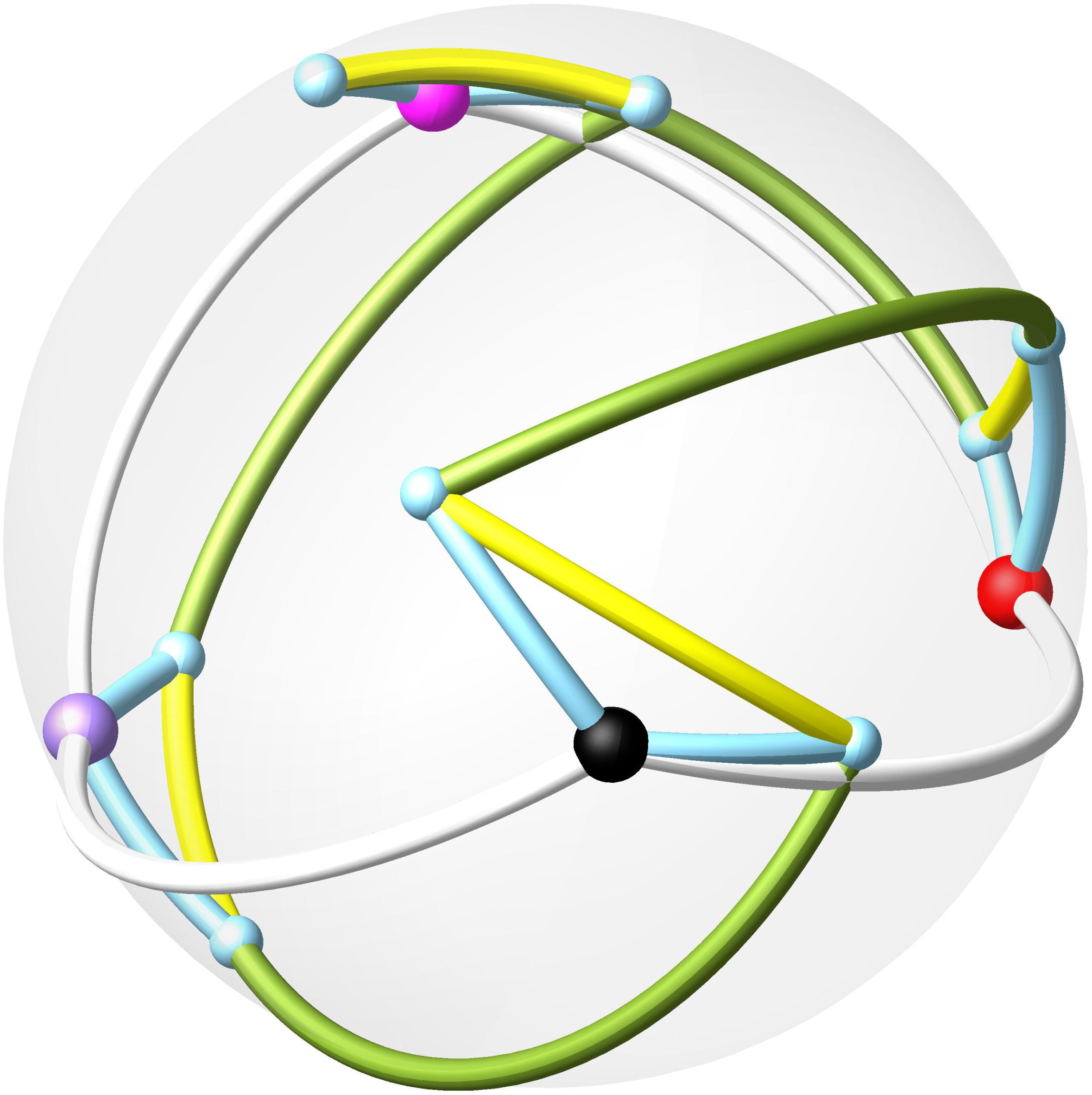}
\begin{scriptsize}
\put(0,0){b)}
\end{scriptsize}     
  \end{overpic} 
	
\caption{A $(3\times 3)$ building block of a V-hedra with skew quads (a) and its spherical image (b). 
The vertex $\go V_i$ and the line-segment $\go V_{i-1}\go V_i$ have the same color, where $i=0$ corresponds to black, $i=1$ to red, $i=2$ to magenta and $i=3$ to purple. Note that the illustrations of Figs.\ \ref{fig4},  
\ref{fig5} and \ref{fig6} are rendered with respect to the same view. 
The animations of the mobility of the V-hedra and its spherical image are online available at 
\url{https://www.dmg.tuwien.ac.at/nawratil/skew_quad_spatial.gif} and 
\url{https://www.dmg.tuwien.ac.at/nawratil/skew_quad_spherical.gif}, respectively. 
}
  \label{fig4}
\end{figure}

\begin{figure}[h!]
\begin{center}
\begin{overpic}
    [height=55mm]{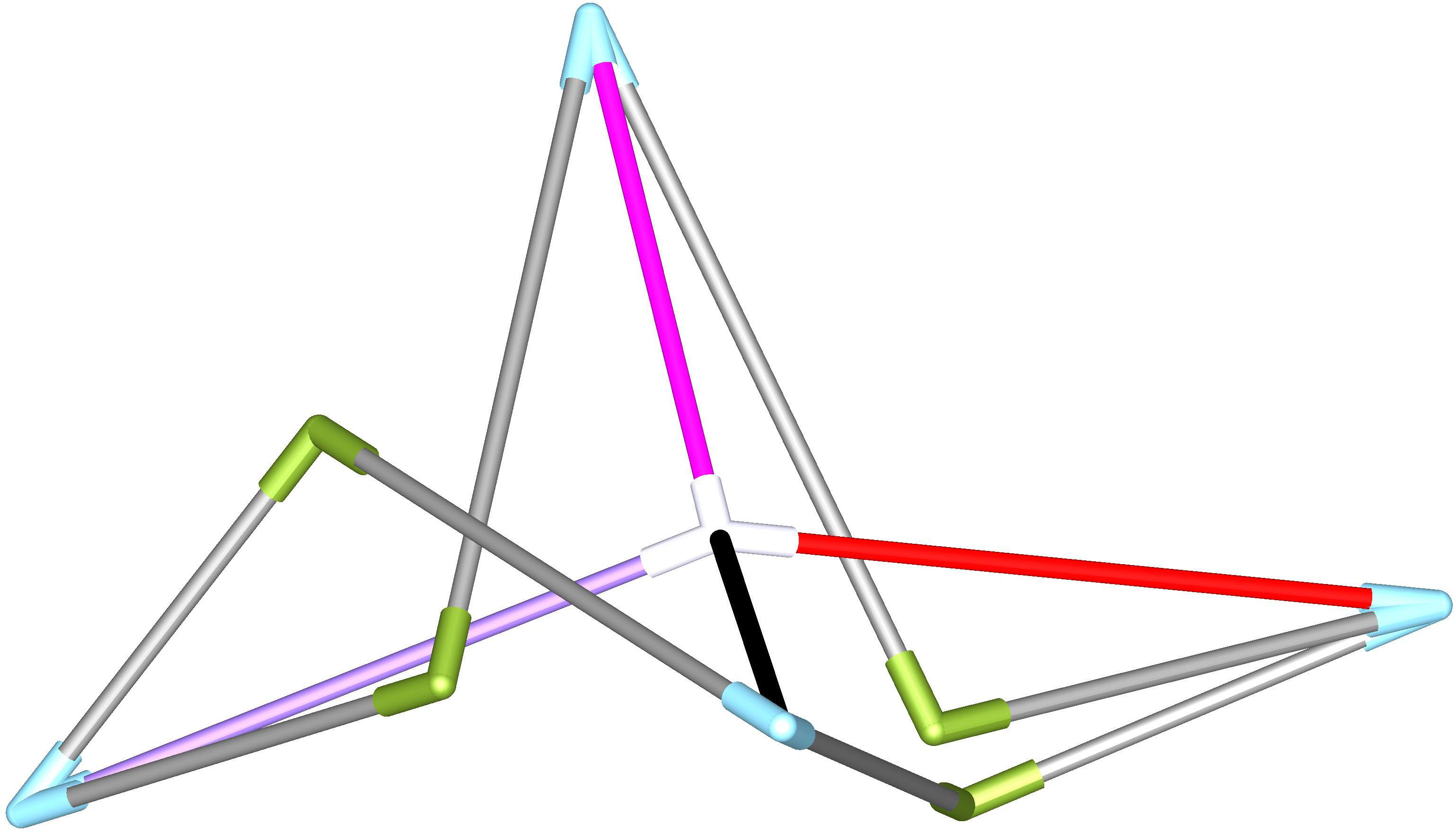}
\end{overpic} 
\end{center}
\caption{The overconstrained mechanism which results from the reciprocal-parallelism to the 
$(3\times 3)$ building block of a V-hedra with skew quads illustrated in Fig. \ref{fig4}a. 
The animation of the mobility of this mechanism is online available at 
\url{https://www.dmg.tuwien.ac.at/nawratil/skew_quad_reciprocal.gif}, where we fixed  
the length of the black edge, which is parallel to $\go V_3\go V_0$.
}
  \label{fig5}
\end{figure}

%%%%%%%%%%%%%%%%%%%%%%%%%%%%%%%%%%%%%%%%%%%%%%%%%%%%%%%%%%%%%%%%%%%%%%%%%%%%%%%%%%%%%%%%%%%%%%%%
%%%%%%%%%%%%%%%%%%%%%%%%%%%%%%%%%%%%%%%%%%%%%%%%%%%%%%%%%%%%%%%%%%%%%%%%%%%%%%%%%%%%%%%%%%%%%%%%
%%%%%%%%%%%%%%%%%%%%%%%%%%%%%%%%%%%%%%%%%%%%%%%%%%%%%%%%%%%%%%%%%%%%%%%%%%%%%%%%%%%%%%%%%%%%%%%%

\newpage

\begin{figure}[t]
\begin{center}
\begin{overpic}
    [width=45mm]{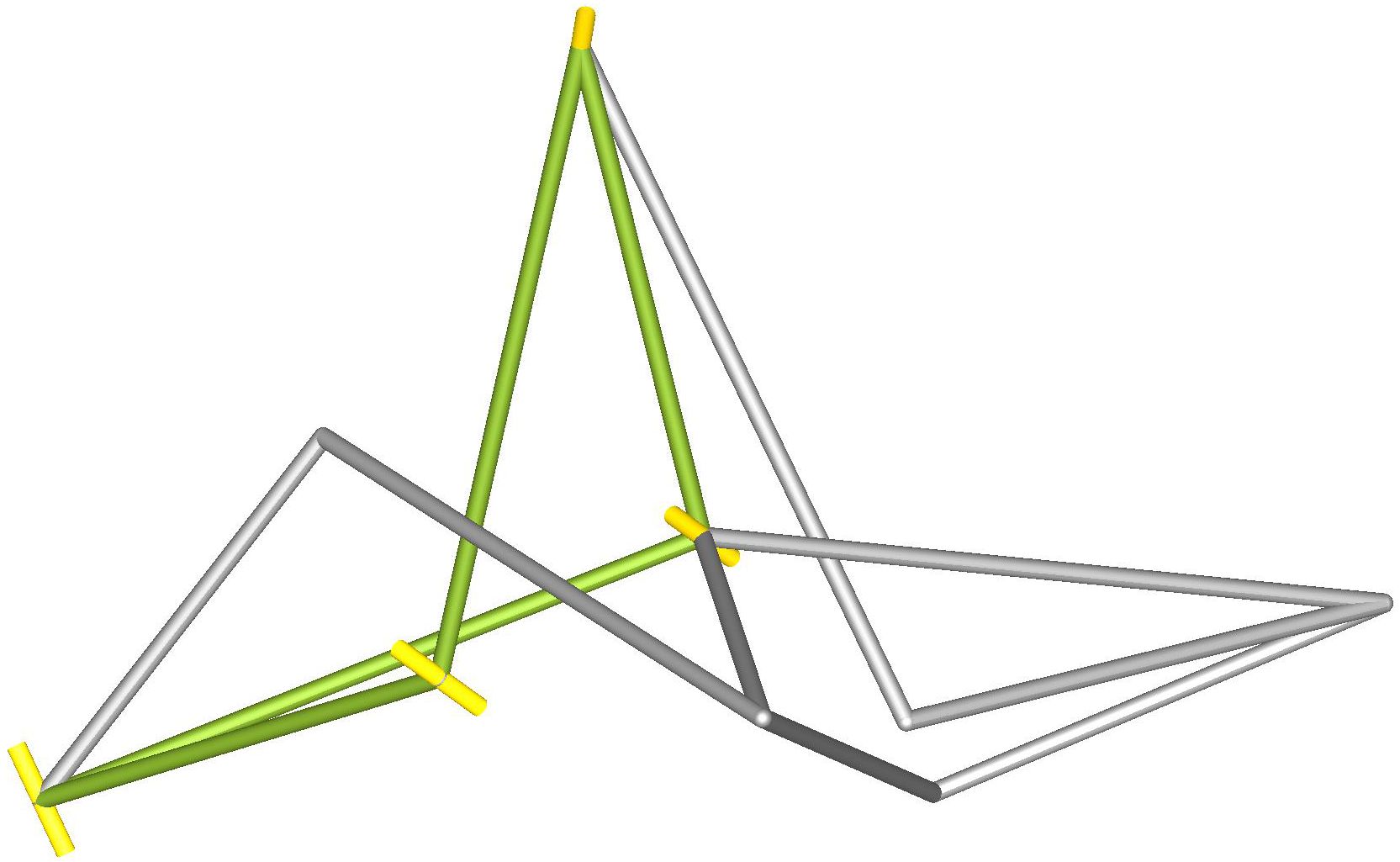}
  \end{overpic} 
\qquad
\begin{overpic}
    [width=45mm]{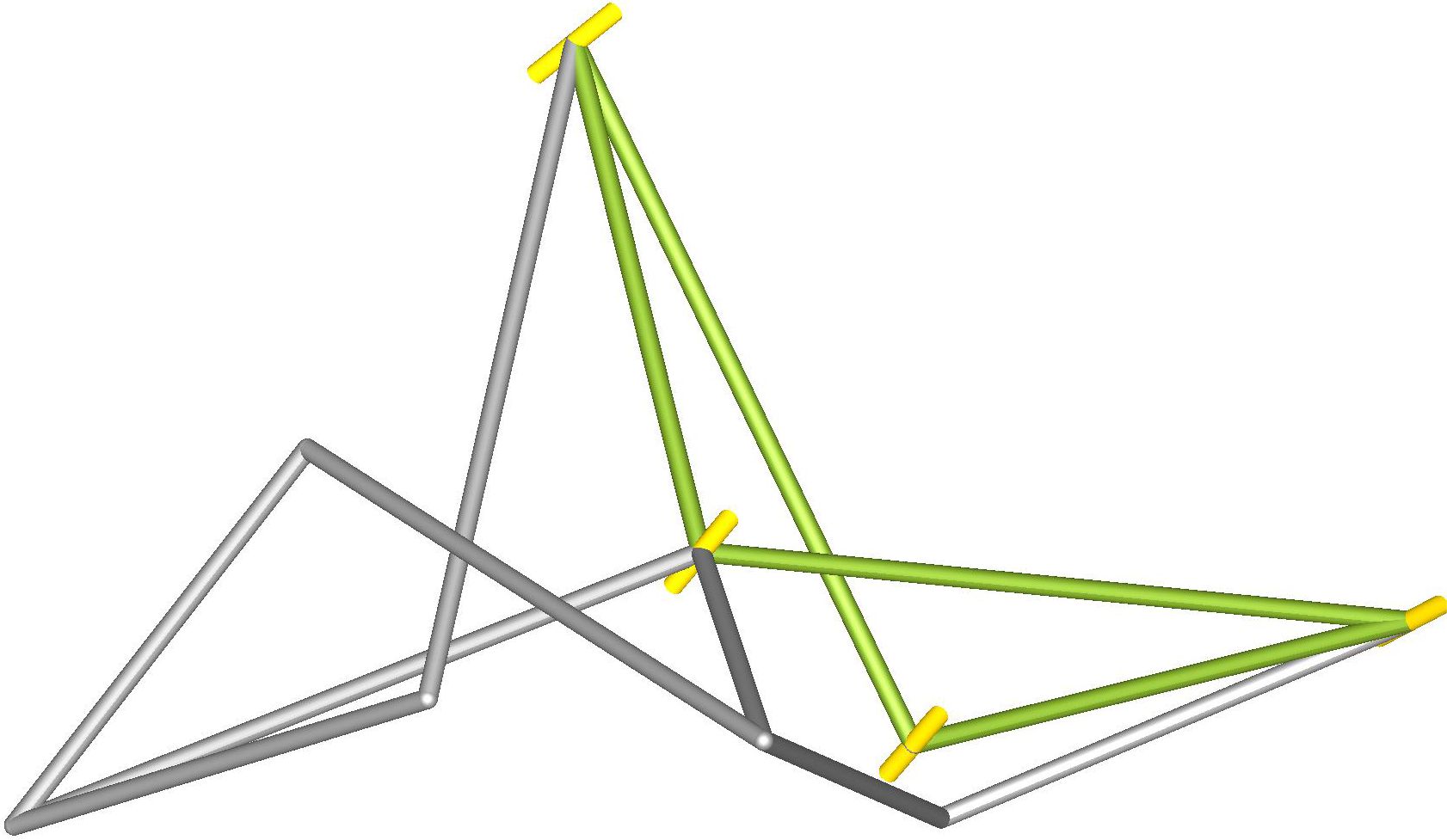}
  \end{overpic} 
\newline
\begin{overpic}
    [width=45mm]{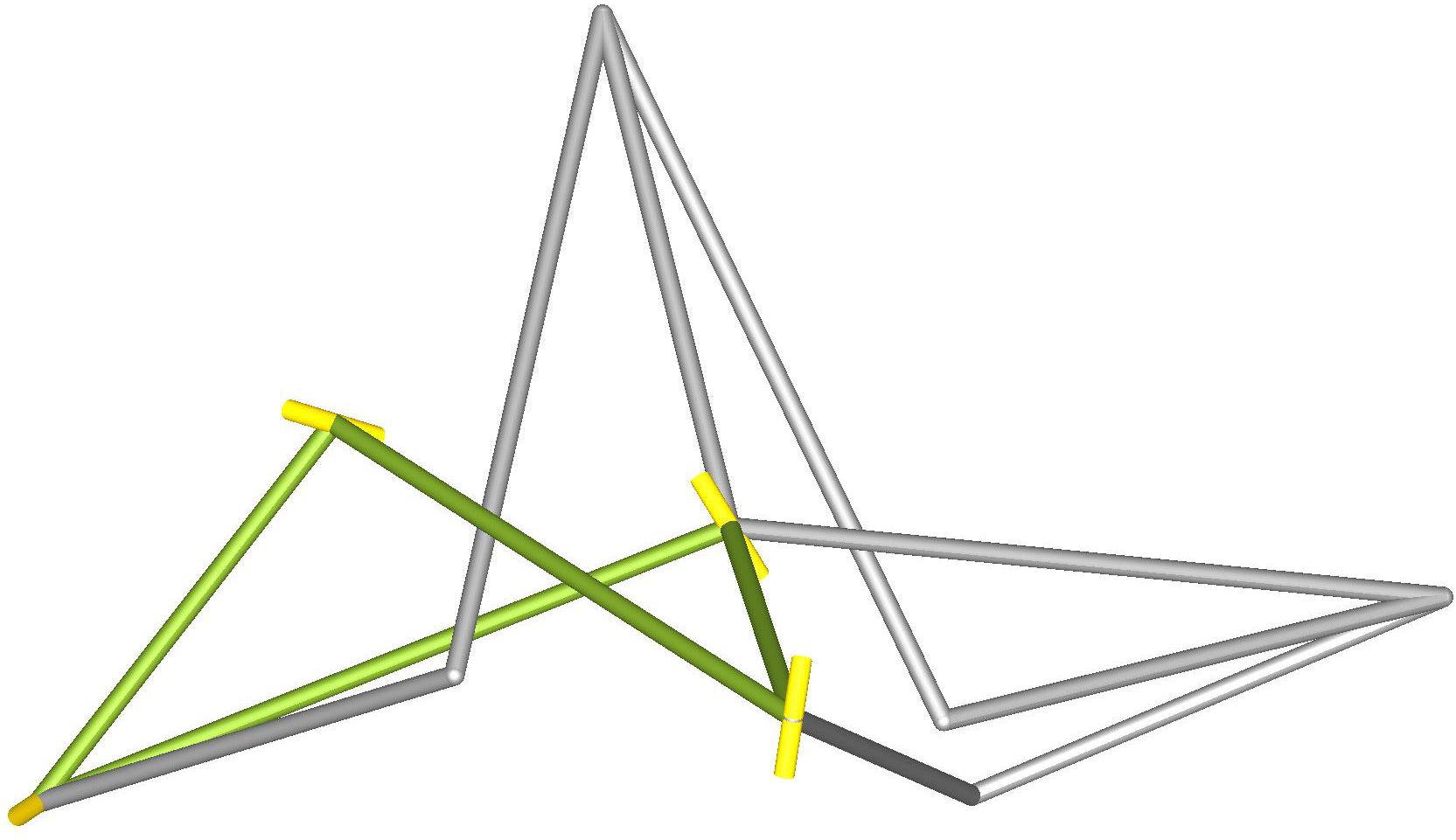}
  \end{overpic} 
\qquad
\begin{overpic}
    [width=45mm]{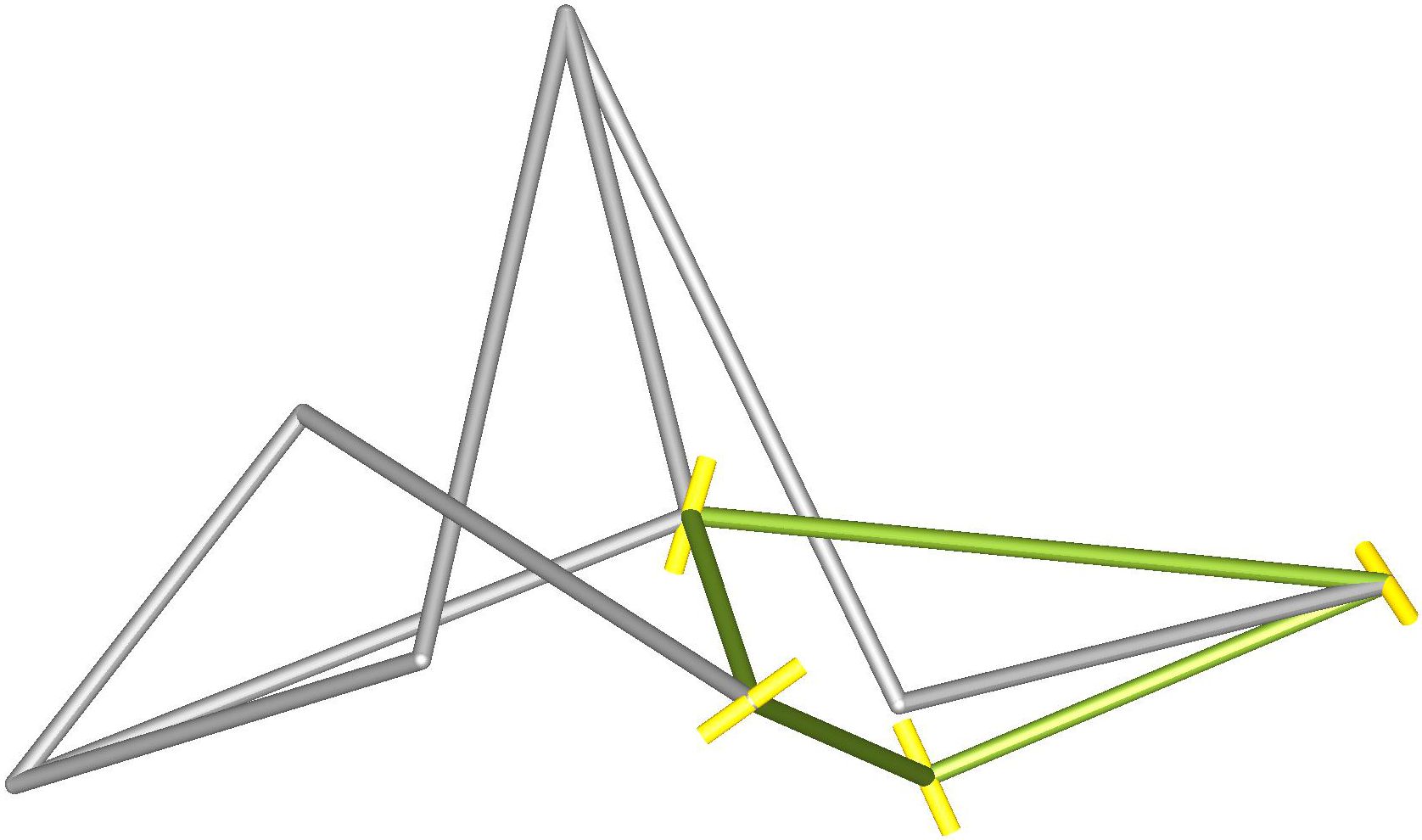}
  \end{overpic} 
\end{center}
\caption{The four Bennett mechanisms associated with the structure illustrated in Fig.\ \ref{fig5}, 
where the rotation axes are displayed in yellow.
}
  \label{fig6}
\end{figure}    

\section{Conclusions, open problems and future research}\label{sec:conclusion}

We generalized continuous flexible Kokotsakis belts of the isogonal type by allowing that 
the faces, which are  adjacent to the line-segments of the rigid closed polygon $\go p$, to be skew. 
In more detail we studied the case where $\go p$ is a skew quad as it corresponds to a 
$(3\times 3)$ building block of a V-hedra composed of skew quads, which proves (under consideration 
of Theorem \ref{th:schief}) the existence of continuous flexible SQ surfaces. 

Open questions in this context regard the smooth analog of continuous flexible Kokotsakis 
belts of the isogonal type and of V-hedra with skew quads.

This study at hand is also the starting point towards a full classification of continuous flexible 
$(3\times 3)$ SQ building blocks, which is subject to future research.

\subsubsection{Acknowledgements.} The research is supported by grant  F77 (SFB ``Advanced Computational Design'', subproject SP7) of the Austrian Science Fund FWF.


\end{document}